\documentclass[conference]{IEEEtran}
\usepackage{mathpazo}
\usepackage{times}
\usepackage{color}
\usepackage{amsmath}
\usepackage{amsfonts}
\usepackage{latexsym}
\usepackage{amssymb}
\usepackage{cite}

\usepackage{upref}
\usepackage{theorem}
\usepackage{graphicx}
\usepackage{psfrag}

\theoremstyle{plain}
\theorembodyfont{\normalfont\slshape}

\newtheorem{thm}{Theorem$\!$}
\newenvironment{theorem}
{\begin{thm}\hspace*{-1ex}{\bf.}}{\end{thm}}

\newtheorem{lem}[thm]{Lemma$\!$}
\newenvironment{lemma}{\begin{lem}\hspace*{-1ex}{\bf.}}{\end{lem}}

\newtheorem{prop}[thm]{Proposition$\!$}

\newtheorem{cor}[thm]{Corollary$\!$}

\newtheorem{defn}[thm]{Definition$\!$}

\newtheorem{xmpl}[thm]{Example$\!$}
\newenvironment{example}{\begin{xmpl}\hspace*{-1ex}{\bf.}}{\hfill$\Box$\end{xmpl}}

\newtheorem{cnstr}{Construction$\!$}

\newcounter{enumrom}
\renewcommand{\theenumrom}{(\roman{enumrom})}

\makeatletter
\renewcommand{\@endtheorem}{\endtrivlist}
\makeatother

\makeatletter
\renewcommand{\thefigure}{{\@arabic\c@figure}}
\renewcommand{\fnum@figure}{{\bf Figure\,\thefigure}}
\makeatother

\newcommand{\cE}{\mathcal{E}}
\newcommand{\cF}{\mathcal{F}}
\newcommand{\cG}{\mathcal{G}}

\newcommand{\cV}{\mathcal{V}}

\newcommand{\mathset}[1]{\left\{#1\right\}}
\newcommand{\abs}[1]{\left|#1\right|}

\newcommand{\parenv}[1]{\left( #1 \right)}

\newcommand{\be}[1]{\begin{equation}\label{#1}}
\newcommand{\ee}{\end{equation}}

\renewcommand{\leq}{\leqslant}

\renewcommand{\geq}{\geqslant}

\renewcommand{\Bbb}{\mathbb}

\newcommand{\Cref}[1]{Co\-ro\-lla\-ry\,\ref{#1}}

\renewcommand{\Bbb}{\mathbb}

\newcommand{\R}{{\Bbb R}}
\newcommand{\Z}{{\Bbb Z}}

\DeclareMathOperator{\inc}{In}
\DeclareMathOperator{\out}{Out}
\DeclareMathOperator{\val}{val}
\DeclareMathOperator{\tval}{\widetilde{val}}
\newcommand{\dinmax}{\delta}
\newcommand{\amax}{\alpha}

\outer\def\proclaim #1. #2\par{\medbreak
 \noindent{\bf#1.\enspace}{\sl#2\par}%
 \ifdim\lastskip<\medskipamount \removelastskip\penalty55\medskip\fi}

\mathchardef\inn="3232
\renewcommand{\in}{{\,\inn\,}}

\begin{document}


\title{\Huge\bf Quasi-linear Network Coding}

\author{
\IEEEauthorblockN{\textbf{Moshe Schwartz}}
\IEEEauthorblockA{Electrical and Computer Engineering\\
Ben-Gurion University of the Negev\\
Beer Sheva 8410501, Israel\\
{\it schwartz@ee.bgu.ac.il}\vspace{-2em}}
\and
\IEEEauthorblockN{\textbf{Muriel M\'edard}}
\IEEEauthorblockA{Research Laboratory of Electronics\\
Massachusetts Institute of Technology\\
Cambridge, MA 02139, USA\\
{\it medard@mit.edu}\vspace{-2em}}
}

\maketitle

\begin{abstract}
We present a heuristic for designing vector non-linear network codes
for non-multicast networks, which we call quasi-linear network codes.
The method presented has two phases: finding an approximate linear
network code over the reals, and then quantizing it to a vector
non-linear network code using a fixed-point representation. Apart from
describing the method, we draw some links between some network
parameters and the rate of the resulting code.
\end{abstract}

\section{Introduction}
\label{sec:intro}

Network coding was introduced in \cite{AhlCaiLiYeu00} as a means of
increasing the amount of information flowing through a network. In
this scheme, a network is a directed graph, where information is
generated by source nodes, and demanded by terminal nodes. All
participating nodes receive information through their incoming edges,
combine the information, and transmit it over their outgoing edges.

\emph{Linear} network coding has drawn particular interest due its
simplicity and structure. Works such as
\cite{AhlCaiLiYeu00,KoeMed03,HoMedKoeKarEffShiLeo06,JagSanChoEffEgnJaiTol05}
studied fundamental bounds on the parameters of such codes, mainly for
linear multicast networks. For these networks, necessary and
sufficient conditions for the existence of a linear solution are
known, as are efficient algorithms for finding such a solution. It was
also shown in \cite{LiYeuCai03}, that solvable multicast networks
always have a scalar linear solution. 

In the non-multicast case, the picture is more complicated. Several
works \cite{RasLeh04,MedEffKarHo03,Rii04} showed various restrictions
on the ability to find a linear solution to these networks. These
culminated in \cite{DouFreZeg05}, that showed linear network codes are
insufficient for solving a non-multicast network over finite fields,
commutative rings, and even $R$-modules. This was done using a single
network, for which no linear solution exists, though an ad-hoc
non-linear solution is possible. Certain classes of networks are known
to have linear solutions over finite fields for the non-multicast
case. The necessity, for the existence of a scalar linear solution
over a finite field, of the matroidal nature of a network was shown in
\cite{DouFreZeg07}, and the sufficiency proven in
\cite{KimMed10}. Recently, \cite{MurSun13} established the equivalence
between the discrete polymatroidal nature of a network and the
existence of a vector linear solution over a finite field for all
feasible non-multicast connections over that network.

The goal of this paper is to introduce \emph{quasi-linear network
  coding}, which is a heuristic method for designing a vector
non-linear network code for non-multicast networks. The method is
inspired by work on real network coding of
\cite{ShiKatJagDeyKatMed08}. The method we suggest has two main
phases.  In the first one, an approximate solution for the network is
found over the reals. We say this is an \emph{approximate} solution,
since at the terminals, the original messages are not recovered
exactly, and there is some mixing with unwanted messages. At the
second stage, restrictions over possible source messages, together
with a fixed-point representation, enable the terminals to reconstruct
the demanded source messages with zero error. We thus gain from both
worlds: the method is linear at its core, giving it some structure,
while at the second phase non-linearity is introduced in a systematic
way, overcoming the insufficiency of linear solutions to non-multicast
networks.

Comparing this work with \cite{ShiKatJagDeyKatMed08}, we note that
both works try to solve the network over the reals. However,
\cite{ShiKatJagDeyKatMed08} considers only the multicast case, and it
assumes an exact solution exists. Furthermore, the real coefficients
in \cite{ShiKatJagDeyKatMed08} are used to obtain a graceful
degradation. This is inherently different in the non-multicast case
that we consider, since an exact solution is not guaranteed, in which
case we use the real coefficients to obtain an approximate solution.

This paper is organized as follows. In Section \ref{sec:pre} we
introduce the required definition and notation. In Section
\ref{sec:method} we describe quasi-linear network codes,
discuss some of their properties, and give examples. We conclude in
Section \ref{sec:conc} with a brief summary and some open questions.

\section{Preliminaries}
\label{sec:pre}

For the purpose of this work, a \emph{network} is a directed acyclic
graph $\cG=(\cV,\cE)$. For a vertex $v\in \cV$ we denote its incoming
edges as $\inc(v)$, and its outgoing edges as $\out(v)$. The in-degree
of vertex $v\in\cV$ is defined as $\delta(v)=\abs{\inc(v)}$, and the
maximal in-degree in the graph is denoted as
\[\dinmax=\max_{v\in\cV}\delta(v).\]

The \emph{depth} of the graph, denoted $d$, is defined as the length
of the longest path in the graph. Since the graph is acyclic, the
depth is well defined.

Two distinguished subsets of vertices are the \emph{source nodes} and
the \emph{terminal nodes}, denoted by $S,T\subseteq\cV$
respectively. The source nodes generate \emph{messages}, which are
symbols from some finite alphabet $\Sigma$. With each of the terminal
nodes we associate a demand for a subset of the messages.

Information is transmitted over edges in the form of symbols from
$\Sigma$.  We denote this transmitted information as $\val(e)\in
\Sigma$ for each $e\in\cE$.  Apart from the source nodes, each node
$v\in\cV$, transmits information along its outgoing edges, which is a
function of the information received on the incoming edges to
$v$. More precisely, for each $v\in \cV\setminus S$, and for each
$e\in\out(v)$,
\[\val(e)=f_e\parenv{ e_1,e_2,\dots,e_{\delta(v)} },\]
where $\inc(v)=\{e_1,\dots,e_{\delta(v)}\}$, and $f_e$ is a function
associated with the edge $e$. Here we implicitly assume a fixed order
of the edges in $\inc(v)$, since $f_e$ is not necessarily a symmetric
function.

Throughout this paper we assume a single source node, i.e.,
$S=\mathset{s}$. The terminal nodes are denoted
$T=\mathset{t_1,t_2,\dots,t_\ell}$. To avoid trivialities, since the
graph is acyclic and there is only one source node, we can assume $s$
is the only vertex in $\cG$ with no incoming edges.  Thus, the
outgoing edges of $s$ transmit the messages of $s$.

The case we study is a general \emph{non-multicast network}, where each
of the terminal nodes demands a subset of the source messages. The subsets
are not necessarily disjoint.

We say a vertex $v\in\cV$ is of depth $d(v)$ if the \emph{longest}
path from $s$ to $v$ is of length $d$. In a similar manner, the depth
of an edge $e\in\cE$, $e=v\to v'$, is defined as the depth of $v$, and
we denote $d(e)=d(v)$. We can now partition the edge set
\[\cE=\cE_0\cup\cE_1\cup\dots\cup\cE_{d-1}.\]
By definition, $e\in\cE_{d(e)}$ for each $e\in\cE$. We note that
$\cE_0=\out(s)$. Again, in order to avoid trivialities, we assume
$\cE_{d-1}$ contains at least one edge ending in a terminal node.

\section{Method and Analysis}
\label{sec:method}

The quasi-linear network-coding method we describe is inspired by the
arithmetic network coding of \cite{ShiKatJagDeyKatMed08}. While the latter work considered  multicast setting, we consider the general non-multicast which does subsume the multicast case. The main strategy
is given by the following two steps:
\begin{enumerate}
\item
  Initially, instead of using a finite field for the alphabet of
  messages, we use real numbers. Nodes linearly combine the real
  scalars on incoming edges using real coefficients. The coefficients are
  chosen so as to approximate the demands at the terminal nodes.
\item
  Messages from the source are restricted to integers.  The messages
  over the edges are replaced with finite-precision fixed-point
  representations in base $b$. The degree of approximation to the
  demands, calculated in the first step, is used to limit the range of
  integers the source may send. The terminal nodes reverse the linear
  combination and quantize the result to the nearest integer, as an
  estimate to the demanded messages.
\end{enumerate}

We observe that at the terminal nodes there are two sources of noise
that may prevent the recovery of the correct demanded messages. The
first is due to the numerical error accumulating along the paths from
the source to the terminal, caused by the restriction to fixed-point
precision. This occurs even if we assume internal computations within
the node are done with infinite precision\footnote{There is no actual need for infinite precision. We can choose a precision high enough within nodes to make internal node computations irrelevant.}.

The second source of noise at the terminals is due to the
approximation to the demand. Terminal nodes essentially compute a
linear combination of the messages from the source. Ideally, this
combination has a coefficient of $1$ for the demanded message, and a
coefficient of $0$ for each of the other messages. However, such a
solution may not be possible, as was demonstrated in
\cite{DouFreZeg05}. We shall therefore strive to obtain a linear
combination with coefficients close to $1$ and $0$
appropriately. Intuitively, such combinations introduce a ``weak''
version of unwanted messages. By limiting the range of messages the
source transmits, these ``weak'' versions of unwanted messages lead to interference that is removed by quantization. Thus, when the terminal recovers the correct
integer message when quantizing the linear combination to the nearest
integer.

\subsection{Stage I -- Working over $\R$}

We now describe the method in detail, and analyze some of its
properties. Let $\cG=(\cV,\cE)$ be a non-multicast network, with
source node $s$ and terminal nodes $T=\mathset{t_1,\dots,t_\ell}$. Let
$k=\abs{\out(s)}$ be the number of outgoing edges from the source $s$,
and let $m_1,m_2,\dots,m_k\in\Z$ be the messages the source transmits,
each on one of its outgoing edges, $\out(s)=\cE_0$. We conveniently
denote $\out(s)=\mathset{e^s_1,e^s_2,\dots,e^s_k}$, and thus,
$\val(e^s_i)=m_i$ for all $i\in [k]$.

For the first stage of the design of the algorithm, assume that over
the rest of the edges, $\cE\setminus\cE_0$, values from $\R$ are
transmitted. We denote the value transmitted over edge $e\in\cE$ as
$\val(e)$. Every vertex in the graph transmits, over its
outgoing edges, linear combinations of values received from the
incoming edges. To be more precise, for every $v\in\cV$, and for every
$e'\in\out(v)$, the value transmitted over $e'$ is the linear
combination
\[\val(e')=\sum_{e\in\inc(v)}\alpha_{e\to e'}\val(e),\]
for some fixed coefficients $\alpha_{e\to e'}\in \R$. For our
convenience, if $e\to e'$ is not a path in $\cG$, we shall set
$\alpha_{e\to e'}=0$.

As is usually done in network coding, let $T$ be the $\abs{\cE}\times\abs{\cE}$
matrix of the line-graph of $\cG$, with
\[T_{e,e'}=\alpha_{e\to e'}.\]
It is well-known \cite{KoeMed03} that by running the network-coding
system, the value transmitted over any edge $e\in \cE$ is given by
\[\val(e)=\sum_{i=1}^{k} m_i \parenv{\sum_{j=0}^{d} T^j}_{e^s_i,e}\]
where $T^0=I$, the identity matrix.

For the sake of brevity and ease of notation, let us assume each of
the terminals demands a single message from the source. The more
general case is a trivial extension of the case we describe. Say
terminal $t\in T$ demands the $w_t$th source message, i.e., $m_{w_t}$,
where $w_t\in [k]$. To that end, the terminal $t$ chooses $\delta(t)$
real coefficients, $\beta_{t,1},\dots,\beta_{t,\delta(t)}$, each associated
with the $\delta(t)$ incoming edges, denoted $e'_1,\dots,e'_{\delta(t)}$. The
terminal node $t$ then performs the linear combination
\begin{equation}
\label{eq:atterm}
\sum_{r=1}^{\delta(t)} \beta_{t,r} \val(e'_r)=
\sum_{i=1}^{k} m_i \sum_{r=1}^{\delta(t)} \parenv{\sum_{j=0}^{d} T^j}_{e^s_i,e'_r}.
\end{equation}
We denote
\begin{equation}
\label{eq:gamma}
\gamma_{t,i}=\sum_{r=1}^{\delta(t)} \parenv{\sum_{j=0}^{d} T^j}_{e^s_i,e'_r},
\end{equation}
and then we can rewrite \eqref{eq:atterm} as
\begin{equation}
\label{eq:rewrite}
\sum_{r=1}^{\delta(t)} \beta_{t,r} \val(e'_r) = \sum_{i=1}^{k} \gamma_{t,i}m_i .
\end{equation}
Since terminal $t$ demands the $w_t$th message, ideally we would like
to get $\gamma_{t,w_t}=1$ and $\gamma_{t,i}=0$ for all $i\neq
w_t$. This goal may be unattainable, in particular, since we need to
solve this concurrently for all $t\in T$.

We therefore resort to try and find an approximate solution over the reals,
as follows. We define the function
\[\cF = \sum_{t\in T}\parenv{ \parenv{\gamma_{t,w_t}-1}^2 + \sum_{\substack{i\in[\delta(t)] \\ i\neq w_t}}\gamma_{t,i}^2},\]
where $\gamma_{t,i}$ is defined in \eqref{eq:gamma}.  By choosing the
real coefficients $\alpha_{e\to e'}$ and $\beta_{t,i}$, the goal is to
minimize $\cF$. 

As a crude overall measure of approximation, we define $\gamma$ as
\[\gamma = \max_{t\in T}\parenv{\abs{\gamma_{t,w_t}-1} + \sum_{\substack{i\in[\delta(t)] \\ i\neq w_t}}\abs{\gamma_{t,i}}}.\]
Intuitively, $\gamma$ is the maximal magnitude of deviation from the
coefficients being $1$ or $0$ appropriately. We say the real solution
to the network is \emph{exact} if $\gamma=0$. This concludes the first
stage of designing a quasi-linear network code.

\subsection{Stage II -- Working with Fixed-Point Precision}

The goal of the second stage of the design of quasi-linear network
codes, is to quantize all the real numbers transmitted over edges to a
fixed-point presentation in base $b$. By doing so, we enable the
transmission of real values as symbols from a finite alphabet, but
also introduce more noise into the system. In this section we go
through this quantization process, and bound the amount of noise
introduced. This will be helpful in determining the range of possible
messages that can be recovered with zero error at the terminals. It
should be noted that the bounds we give are general but crude, and
that for specific networks, a careful analysis in the spirit of this
section, will provide better bounds.

We denote the maximal linear-combination coefficient, chosen in the
previous stage, as
\[\amax=\max_{e,e'\in\cE}\abs{\alpha_{e\to e'}}.\]
We assume $\alpha > 1$. If that is not the case, either the solution
is a trivial routing, or we can scale all the coefficients and receive
a scaled version of the result at the terminals.  We also assume that
the magnitude of all source messages is upper bounded by
\[ \abs{m_i}\leq M.\]
Using our partition of the edges by depth, let us denote
\[M_i=\max_{e\in \cE_i}\abs{\val(e)}.\]

\begin{lemma}
\label{lem:mmax}
For all $0\leq i\leq d-1$ we have
\[M_i\leq (\dinmax\amax)^i M.\]
\end{lemma}
\begin{IEEEproof}
This is a simple proof by induction. For the induction base we have
for each $e\in\cE_0$
\[\abs{\val(e)}\leq M = (\dinmax\amax)^0 M.\]
Since this is true for all edges $e\in\cE_0$ we have
\[M_0\leq (\dinmax\amax)^0 M.\]

For the induction step, let $e'=v\to v'$ be some edge, $e'\in
\cE_i$. Then
\begin{align*}
\abs{\val(e')} &= \abs{\sum_{e\in\inc(v)}\alpha_{e\to e'}\val(e)}
\leq \sum_{e\in\inc(v)}\abs{\alpha_{e\to e'}\val(e)}\\
& \leq \sum_{e\in\inc(v)}\amax M_{d(e)}
\leq \sum_{e\in\inc(v)}\amax (\dinmax\amax)^{d(e)}M\\
& \leq \sum_{e\in\inc(v)}\amax (\dinmax\amax)^{d(e')-1}M
= (\dinmax\amax)^i M,
\end{align*}
where we used the fact that $\amax\dinmax\geq 1$.
Since this holds for any $e\in\cE_i$ we have
\[M_i \leq (\dinmax\amax)^i M.\]
\end{IEEEproof}

Suppose now every edge can carry a value in fixed-point base-$b$
representation with $P$ digits left of the fixed point, and $p$ digits
to the right of it. The nodes calculate the same linear
combinations on their inputs as before, and we assume infinite precision
within the nodes. However, before transmitting the results over the outgoing
messages, a quantization occurs.

We denote the value of the fixed-point representation sent over edges
as $\tval(e)$. For all $0\leq i\leq d-1$ we denote
\[\epsilon_i=\max_{e\in\cE_i}\abs{\tval(e)-\val(e)}.\]

\begin{lemma}
\label{lem:epsmax}
For all $0\leq i\leq d-1$ we have
\[\epsilon_i\leq 
\frac{\parenv{\dinmax\amax}^i-1}
{\dinmax\amax-1}b^{-p}.
\]
\end{lemma}
\begin{IEEEproof}
For convenience, let us denote the RHS of the claim as $f(i)$. We
observe that $f(i)$ is non-negative and monotone increasing in $i$.

The proof is by induction. For the induction base we note that
$\epsilon_0=0$ since, in our setting, the source node transmits only
integer values, and we shall make sure to set $P$ to a large enough
value so that no truncation error occurs.

For the induction step consider any edge $e'\in \cE_i$. We now have
\begin{align*}
&\abs{\tval(e')-\val(e')} =\\
&\quad = b^{-p} +\abs{\sum_{e\in\inc(v)}\alpha_{e\to e'}\tval(e)-\alpha_{e\to e'}\val(e)}\\
&\quad \leq b^{-p}+ \sum_{e\in\inc(v)}\abs{\alpha_{e\to e'}\epsilon_{d(e)}}\\
&\quad \leq b^{-p}+ \sum_{e\in\inc(v)}\abs{\alpha_{e\to e'}f(d(e))}\\
&\quad \leq b^{-p}+ \sum_{e\in\inc(v)}\amax f(i-1)\\
&\quad \leq b^{-p}+ \dinmax\amax f(i-1)\\
&\quad = f(i).
\end{align*}
Since this holds for any $e'\in\cE_i$ we have 
\[\epsilon_i\leq f(i).\]
\end{IEEEproof}

We are now in a position to combine all of the previous observations,
and give sufficient conditions for a zero-error recovery of the values
at the terminals.

\begin{theorem}
\label{th:prec}
Using the quasi-linear network-coding scheme describe before, it is
possible to recover the original values at the terminals if
\begin{align*}
  M & < \frac{1}{2\gamma} \qquad \text{(only if $\gamma>0$),}\\
  p & >  \log_b \frac{\parenv{\dinmax\amax}^{d-1}-1}
{\dinmax\amax-1} - \log_b\parenv{\frac{1}{2}-\gamma M}, \\
  P &\geq \log_b \parenv{2(\dinmax\amax)^{d-1}M+2}.
\end{align*}
\end{theorem}
\begin{proof}
We first note that at any terminal $t\in T$, the absolute difference between
the demanded message and the linear combination obtained by
\eqref{eq:rewrite} and Lemma \ref{lem:epsmax}, is upper bounded by
\[\abs{m_{w_t}-\sum_{i=1}^{k} \gamma_{t,i}m_i}+\epsilon_{d-1}\leq
\gamma M+ \frac{\parenv{\dinmax\amax}^{d-1}-1}
{\dinmax\amax-1}b^{-p}.\]
Thus, if we require
\[\gamma M+ \frac{\parenv{\dinmax\amax}^{d-1}-1}
{\dinmax\amax-1}b^{-p} < \frac{1}{2},\]
then rounding the resulting linear combination at terminal $t$ to the
nearest integer, will recover the message correctly.

It follows that when $\gamma>0$, i.e., the real solution in the first stage
is not exact, we must require
\[\gamma M< \frac{1}{2}.\]
Furthermore, after rearranging and solving for $p$, we obtain the
desired requirement,
\[p >  \log_b \frac{\parenv{\dinmax\amax}^{d-1}-1}
{\dinmax\amax-1} - \log_b\parenv{\frac{1}{2}-\gamma M}.\]

Furthermore, according to Lemma \ref{lem:mmax}, the maximal value sent
on an edge is upper bounded by
\[M_{d-1}=(\dinmax\amax)^{d-1}M.\]
Since, in the previous paragraphs, we bounded the quantization error
on any edge by $\frac{1}{2}$, and since we need the integers in the
range $[-M_{d-1}-\frac{1}{2},M_{d-1}+\frac{1}{2}]$, taking $P \geq
\log_b \parenv{2(\dinmax\amax)^{d-1}M+2}$ digits to the left of the
fixed point ensures the value is within the representation range.
\end{proof}

We briefly pause to contemplate the implications of Theorem
\ref{th:prec}.  If, in the first stage, we are able to find an
\emph{exact} real solution, i.e., with $\gamma=0$, then there is no
bound on the magnitude of the source messages sent. It follows that,
in this case, the number of digits used when transmitting over any
edge, $P+p$, is $\log_b M+ c$, where $c$ is some constant that depends
on the network topology. Thus, we can get quasi-linear network-coding
solutions with rate arbitrarily close to $1$, where rate is measured
as the ration between the minimum number of bits required to describe
a source message, and the number of bits used for transmission over
any edge.

When we do not have an exact real solution in the first stage, we can
no longer support arbitrarily-large source messages. Furthermore, we
note that the bound on $p$ from Theorem \ref{th:prec} distinctly
exhibits a component affected by the degree of approximation $\gamma$,
and a component affected by the fixed-point quantization.

\begin{example}
The network $\cG_1$, shown in Figure \ref{fig:g1}, was given in
\cite{DouFreZeg05} as part of a larger network. The network has a
source node $s$ that produces five source messages
$m_1,\dots,m_5$. There are seven terminal nodes, with single-message
demands written beneath the appropriate node.

It was shown in \cite{DouFreZeg05} that $\cG_1$ has no scalar linear
network-coding solution over $GF(2^h)$, for any $h$. If we restrict
ourselves to a routing solution, in which nodes cannot linearly
combine incoming messages, the best rate we can achieve is
$\frac{1}{3}$, as was shown in \cite{CanDouFreZeg06}.

However, the network is exactly solvable over $\R$ in the following
simple manner. All source nodes repeat their message on all outgoing
edges. All internal nodes sum all incoming messages and transmit the
sum over all the outgoing edges. The terminal nodes perform simple
subtraction to obtain their demands, except for the middle terminal
node, $v_4$, which computes $\frac{1}{2}(-v_1+v_2+v_3)$. It easily
follows that given source messages of $n$ bits, we can find a
quasi-linear network-coding solution with $P=n+2$ and $p=0$ digits to
the left and to the right of the fixed point, respectively. Since $n$
is arbitrarily large, the achievable rate is $\frac{n}{n+2}$, which is
asymptotically $1$ as $n\to\infty$.
\end{example}

\begin{figure}
\begin{center}
\psfrag{m1}{$m_1$}
\psfrag{m2}{$m_2$}
\psfrag{m3}{$m_3$}
\psfrag{m4}{$m_4$}
\psfrag{m5}{$m_5$}
\psfrag{v1}{$v_1$}
\psfrag{v2}{$v_2$}
\psfrag{v3}{$v_3$}
\psfrag{v4}{$v_4$}
\psfrag{ss}{$s$}
\includegraphics[scale=0.7]{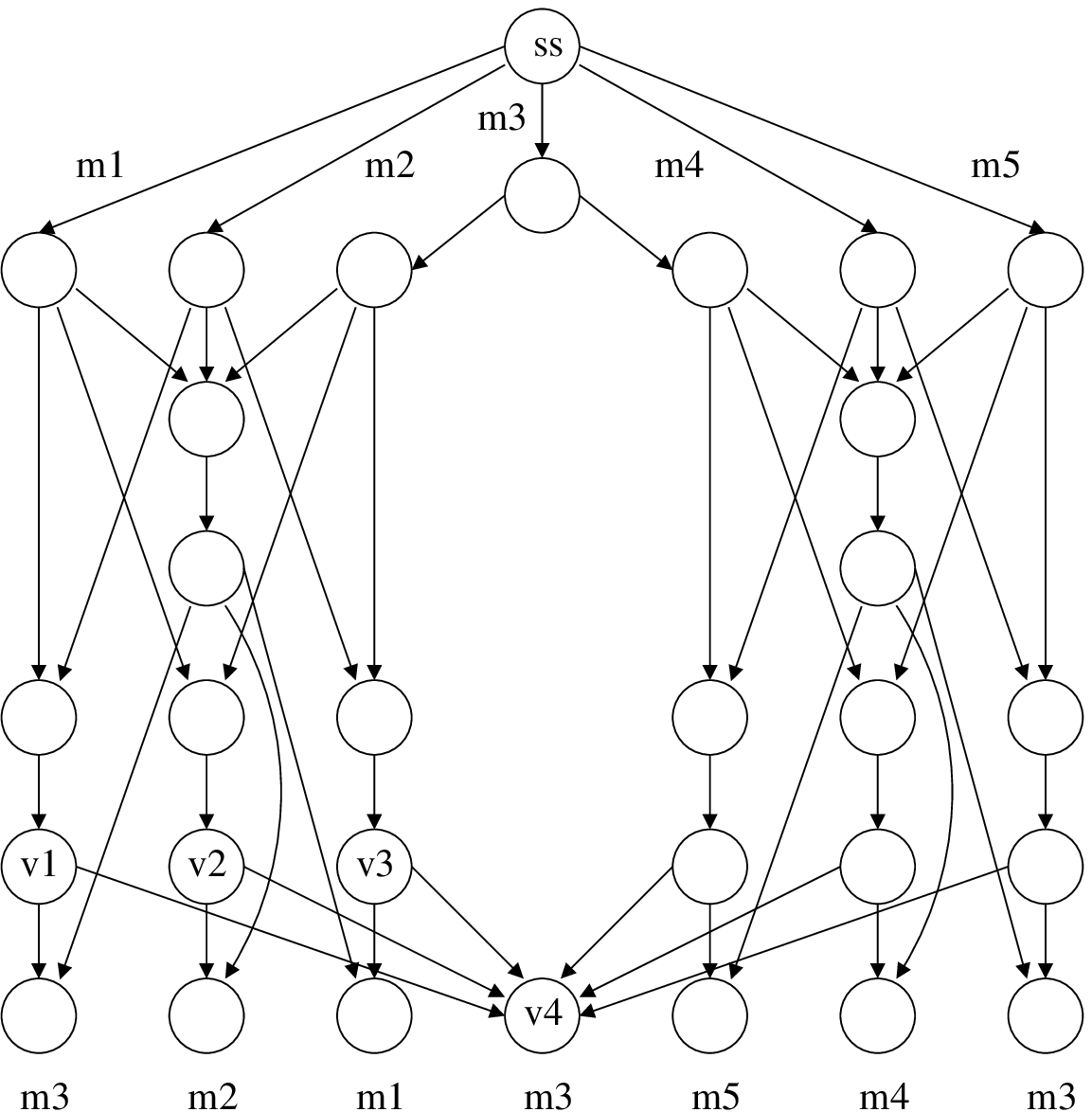}
\end{center}
\caption{The network $\cG_1$.}
\label{fig:g1}
\end{figure}

\begin{example}
A more interesting example is the network $\cG_2$, shown in Figure
\ref{fig:g2}, which was also given in \cite{DouFreZeg05}. The network
has a single source $s$, which produces three message $m_1$, $m_2$,
and $m_3$. The network also has three terminal nodes, whose demands
are written below them.

It was shown in \cite{DouFreZeg05}, that $\cG_2$ has no scalar linear
solution over $GF(q)$ when $q$ is odd, and does have a scalar linear
solution over $GF(2^h)$.

We bring this network as an example for a network that has no exact
real solution. An approximate solution which was found using a computer
search is detailed below:
\begin{align*}
\alpha_{e_1\to e_5} & = 0.0332528  & \alpha_{e_2\to e_5} & = -11.8712 \\
\alpha_{e_3\to e_6} & = 16.3384 & \alpha_{e_4\to e_6} & = 2.69746 \\
\alpha_{e_7\to e_{11}} & = 2.79007 & \alpha_{e_8\to e_{11}} & = 2.02721 \\
\alpha_{e_9\to e_{12}} & = -1.16509 & \alpha_{e_{10}\to e_{12}} & = 2.28349 \\
\beta_1 & = -0.0169705 & \beta_2 & = 0.182872\\
\beta_3 & = -0.030174 & \beta_4 & = 0.0722992\\
\beta_5 & = -25.8106 & \beta_6 & = 21.8495
\end{align*}
where the $\beta$'s are the coefficients used at the terminals to
recover the original messages, and are written next to the edge they
apply to. All the nodes with a single incoming edge simply repeat the
incoming message on all outgoing edges.

The approximation factor turn out to be
\[\gamma=0.00572545,\]
which, by Theorem \ref{th:prec}, allows us to use integers in the
range $[-87,87]$. Since we use base $2$ in this example, we use the
range $[-64,63]$, whose integers may be expressed using $7$ bits.

The other parameters involved in this network are maximum in-degree
$\dinmax=2$, maximum coefficient magnitude $\alpha=16.3384$, and base
$b=2$. Using Theorem \ref{th:prec} again, we find the $P\geq 18$ and
$p\geq 8$ suffice.

However, these estimates for $P$ and $p$ are far from being tight. We
first note that, along any path from source to terminal, there are only
two non-terminal nodes that perform a linear combination. Thus, the
``effective'' depth is only $3$ in this case. Furthermore, by using
the exact values of the various $\alpha_{e\to e'}$, instead of the upper
bound $\alpha$, we can obtain a tighter bound on the maximal value passing
over an edge, and the maximal error $\epsilon_{d-1}$. In this case, after
a simple processing by a computer, these give $P\geq 14$ and $p\geq 6$.

It follows that the network $\cG_2$, using the quasi-linear network
coding system, is capable of transmitting source message of length $7$
bits, using messages of length $14+6=20$ bits, i.e., with rate $7/20 >
1/3$.
\end{example}

\begin{figure}
\begin{center}
\psfrag{m1}{$m_1$}
\psfrag{m2}{$m_2$}
\psfrag{m3}{$m_3$}
\psfrag{ss}{$s$}
\psfrag{b1}{\tiny $\beta_1$}
\psfrag{b2}{\tiny $\beta_2$}
\psfrag{b3}{\tiny $\beta_3$}
\psfrag{b4}{\tiny $\beta_4$}
\psfrag{b5}{\tiny $\beta_5$}
\psfrag{b6}{\tiny $\beta_6$}
\psfrag{e1}{\tiny $e_1$}
\psfrag{e2}{\tiny $e_2$}
\psfrag{e3}{\tiny $e_3$}
\psfrag{e4}{\tiny $e_4$}
\psfrag{e5}{\tiny $e_5$}
\psfrag{e6}{\tiny $e_6$}
\psfrag{e7}{\tiny $e_7$}
\psfrag{e8}{\tiny $e_8$}
\psfrag{e9}{\tiny $e_9$}
\psfrag{ea}{\tiny $e_{10}$}
\psfrag{eb}{\tiny $e_{11}$}
\psfrag{ec}{\tiny $e_{12}$}
\includegraphics[scale=0.7]{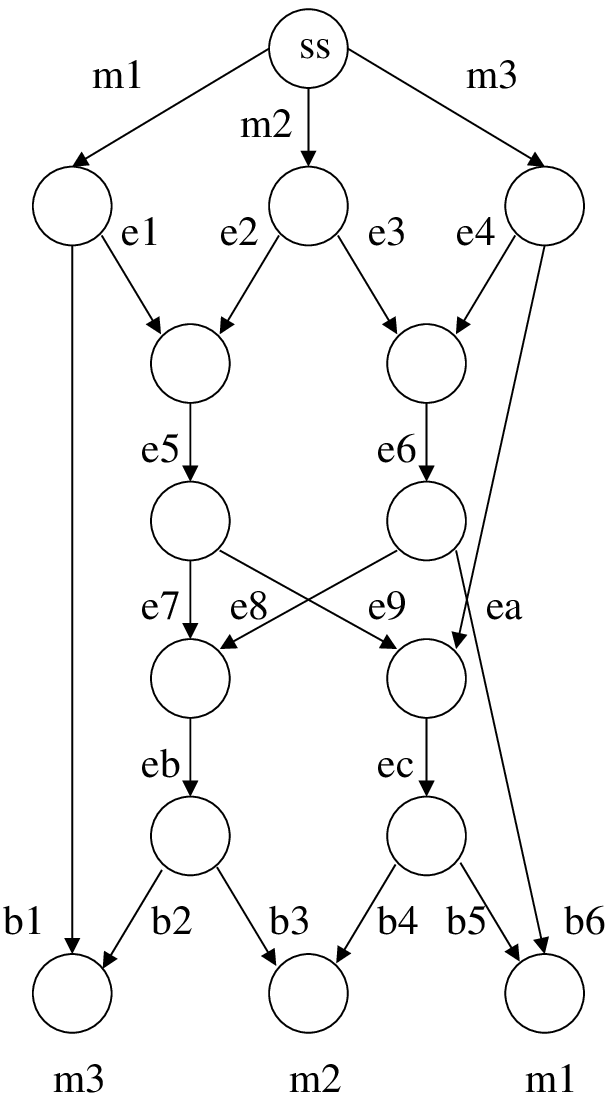}
\end{center}
\caption{The network $\cG_2$.}
\label{fig:g2}
\end{figure}

\begin{example}
As a final example we bring the network $\cG_3$, shown in Figure
\ref{fig:g3}, which is a combination of the networks $\cG_1$ and
$\cG_2$. This combination was shown to have no scalar linear solution
in \cite{DouFreZeg05}, though an ad-hoc non-linear solution was given.

We can apply the quasi-linear network-coding scheme to this network,
and since the two sub-networks operate separately, get a solution with
rate $7/20> 1/3$. We can compare this with the best routing solution
for this network which has a lower rate of $1/3$ (see
\cite{CanDouFreZeg06}).

We are not forced in any way to use the quasi-linear scheme for the
entire network. By combining a routing solution for $\cG_2$ with rate
$2/3$ (see \cite{CanDouFreZeg06}), with a quasi-linear solution for
$\cG_1$ with rate $\frac{n}{n+2}$, for any $n$, we can obtain an
overall solution with rate $2/3$.
\end{example}

\begin{figure*}
\begin{center}
\psfrag{m1}{$m_1$}
\psfrag{m2}{$m_2$}
\psfrag{m3}{$m_3$}
\psfrag{m4}{$m_4$}
\psfrag{m5}{$m_5$}
\psfrag{ss}{$s$}
\includegraphics[scale=0.7]{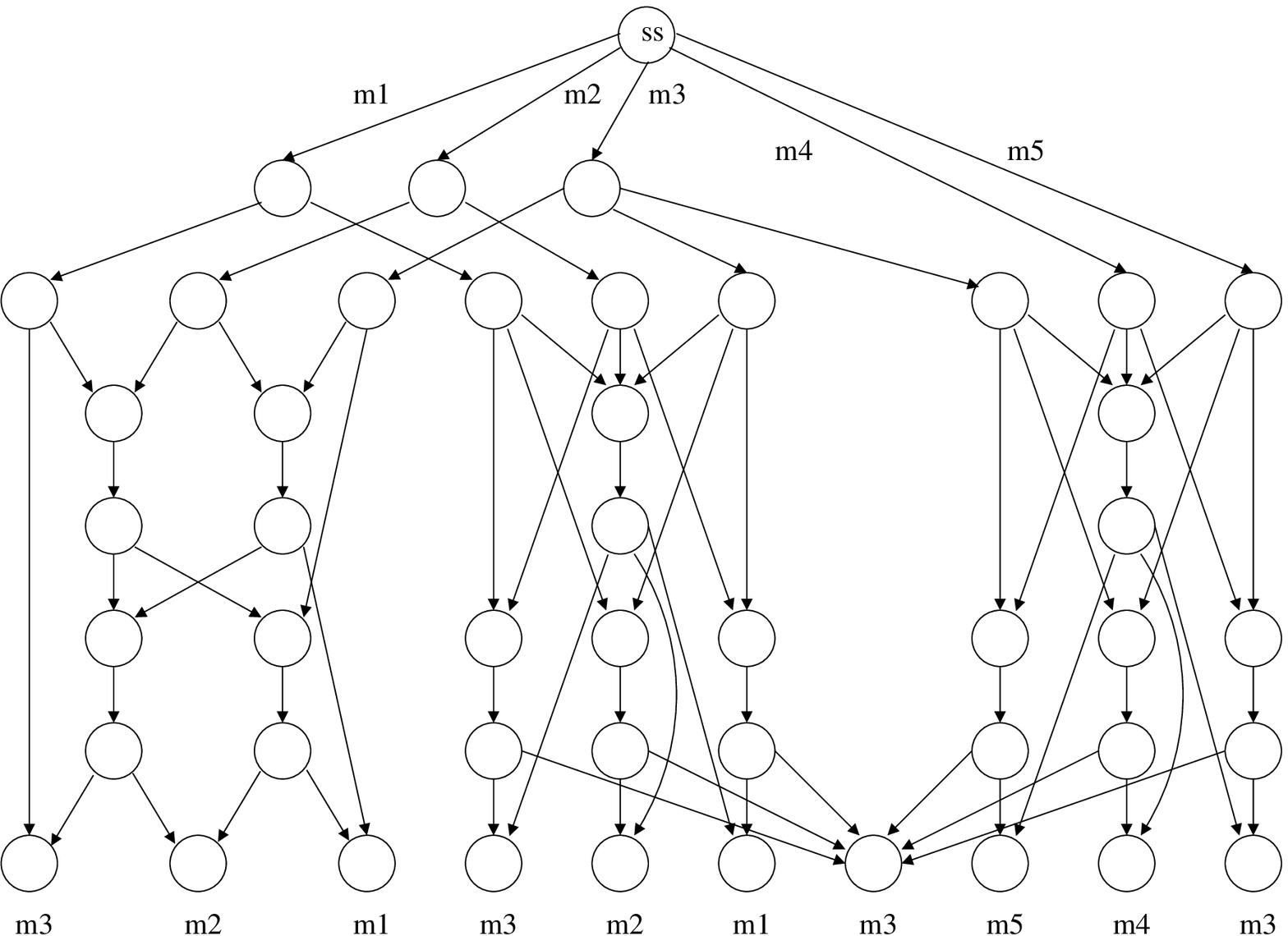}
\end{center}
\caption{The network $\cG_3$.}
\label{fig:g3}
\end{figure*}

\section{Conclusion}
\label{sec:conc}

In this paper we described quasi-linear network coding, which is a
two-phase heuristic for designing vector non-linear network codes for
non-multicast networks. In the first stage an approximate solution
over the reals is found, and in the second stage it is quantized to a
fixed-point representation. We analyzed the sources for errors in the
process and determined sufficient conditions for zero-error at the
terminals. These condition determine the rate of the solution.

We applied the method to the network presented in \cite{DouFreZeg05}
to prove the insufficiency of linear network codes. While overall the
rate was below that of the ad-hoc non-linear solution given in
\cite{DouFreZeg05}, our method is systematic, and it out-performs the
routing capacity of the network.

Connections between the rate of the quasi-linear network code, and
various parameters of the network, e.g., depth and incoming degree,
were established.  However, a crucial piece is still missing, and that
is connecting the approximation factor $\gamma$ with the network. This
missing link will enable us to fully compare quasi-linear network
codes with other coding techniques.

\section*{Acknowledgments}

The first author would like to thank the second author for hosting him
at MIT during his sabbatical.

\bibliographystyle{IEEEtranS}
\bibliography{allbib}

\end{document}